\documentclass[aps,superscriptaddress,pra,notitlepage,nofootinbib,twocolumn,groupedaddress]{revtex4-1}

\usepackage[utf8]{inputenc} % Any characters can be typed directly from the keyboard, eg éçñ
\usepackage{textcomp} % provide lots of new symbols
\usepackage{graphicx}  % Add graphics capabilities
\usepackage{flafter}  % Don't place floats before their definition

\usepackage{amsmath,amssymb,amsfonts,amsthm}  % Better maths support & more symbols
\usepackage{bm}  % Define \bm{} to use bold math fonts
\usepackage{bbm}

\usepackage{verbatim}
\usepackage{xcolor}

\usepackage[pdftex,bookmarks,colorlinks,breaklinks]{hyperref}  % PDF hyperlinks, with coloured links
\definecolor{dullmagenta}{rgb}{0.4,0,0.4}   % #660066
\definecolor{darkblue}{rgb}{0,0,0.4}
\hypersetup{linkcolor=red,citecolor=blue,filecolor=dullmagenta,urlcolor=darkblue} % coloured links
\usepackage{memhfixc}  % remove conflict between the memoir class & hyperref
\usepackage{pdfsync}  % enable tex source and pdf output syncronicity

\newcommand{\gnorm}[2]{\left\|#1\right\|_{#2}}
\newcommand{\dnorm}[1]{\gnorm{#1}{\diamond}}
\newcommand{\cbnorm}[1]{\gnorm{#1}{\rm cb}}
\newcommand{\onenorm}[1]{\gnorm{#1}{1}}
\newcommand{\opnorm}[1]{\gnorm{#1}{\infty}}

\newtheorem{theorem}{Theorem}
\newtheorem{lemma}{Lemma}
\newtheorem{corollary}{Corollary}

\usepackage{braket}
\newcommand{\ketbra}[1]{|#1\rangle\langle #1|}
\def\tr{{\rm Tr}}

\def\eps{\varepsilon}
\def\cA{\mathcal A}
\def\cM{\mathcal M}
\def\cC{\mathcal C}
\def\cE{\mathcal E}
\def\cH{\mathcal H}
\def\cN{\mathcal N}
\def\sB{\mathsf B}
\def\sS{\mathsf S}
\def\sC{\mathtt C}
\renewcommand{\rho}{\varrho}

\def\id{\mathbbm{1}}
\newcommand*{\half}{\frac{1}{2}}

\newcommand*{\cF}{\mathcal{F}}

\newcommand*{\cI}{\mathcal{I}}

\newcommand*{\cQ}{\mathcal{Q}}

\newcommand*{\cP}{\mathcal{P}}

\usepackage{xspace}

\usepackage{tikz}

\usetikzlibrary{positioning}
\usetikzlibrary{fit}

\definecolor{waikawa}{rgb}{0.345,.455,.596}
\definecolor{finlandia}{rgb}{.345,.439,.345}
\definecolor{olivine}{cmyk}{0.173, 0, 0.32, 0.227}
\definecolor{azure}{rgb}{.2,.4,.6}

\definecolor{lightcoral}{rgb}{.941,.5,.5}
\definecolor{cornflowerblue}{rgb}{.392,.584,.929}
\definecolor{lightcfblue}{rgb}{.55,.69,.96}
\definecolor{sand}{rgb}{1,.88,.36}
\definecolor{navajowhite}{rgb}{1,.871,.678}
\definecolor{palegreen}{rgb}{.565,.933,.565}
\definecolor{limegreen}{rgb}{.196,.804,.196}
\definecolor{chartreuse}{rgb}{.5,1,0}
\definecolor{plaingreen}{rgb}{0.5,.75,0.5}
\definecolor{anothergreen}{rgb}{.32,.81,.29}
\definecolor{thistle}{rgb}{.847,.749,.747}
\definecolor{lavender}{rgb}{.902,.902,.980}
\definecolor{honeydew}{rgb}{.941,1,.941}
\definecolor{palegoldenrod}{rgb}{.933,.91,.667}

\definecolor{darksand}{rgb}{1,.64,.36}

\def\figjmbasic{
\begin{tikzpicture}[
	apparatus/.style={draw,fill=palegoldenrod,minimum height=18mm,minimum width=5mm},
	measurement/.style={draw,fill=blue!20}
	]
	\node[apparatus] (app) at (0,0) {\Large $\mathcal A_{X,Z}$};
	\draw[very thick] (app.west) -- node[below] {$S$} ++(-0.5,0);
	\draw ([yshift=5.5mm]app.east) -- ++(0.6,0);
	\draw ([yshift=4.5mm]app.east) -- node[below] {$R_X$} ++(0.6,0);
	\draw ([yshift=-4.5mm]app.east) -- ++(0.6,0);
	\draw ([yshift=-5.5mm]app.east) -- node[below] {$R_Z$} ++(0.6,0);
\end{tikzpicture}
}

\def\figjmX{
\begin{tikzpicture}[scale=0.75,transform shape,
	apparatus/.style={draw,fill=palegoldenrod,minimum height=18mm,minimum width=5mm},
	measurement/.style={draw,fill=blue!20}
	]
	\node[apparatus] (app) at (0,0) {\Large $\mathcal A_{X,Z}$};
	\draw[very thick] (app.west) -- node[below] {$S$} ++(-0.5,0);
	\draw ([yshift=5.5mm]app.east) -- ++(0.6,0);
	\draw ([yshift=4.5mm]app.east) -- node[below] {$R_X$} ++(0.6,0);
	\draw[semitransparent] ([yshift=-4.5mm]app.east) -- ++(0.6,0);
	\draw[semitransparent] ([yshift=-5.5mm]app.east) -- node[below] {$R_Z$} ++(0.6,0);
	
	\node[draw,fill=lightcfblue,minimum size=12mm] (X) at (3.6,0) {\Large $\mathcal Q_X$};
	\draw[very thick] (X.west) -- node[below] {$S$} ++(-0.5,0);
	\draw ([yshift=0.5mm]X.east) -- ++(0.6,0);
	\draw ([yshift=-0.5mm]X.east) -- node[below] {$R_X$} ++(0.6,0);
	
	\node (approx) at (1.8,0) {\Large $\approx$};
	\node at ([xshift=0.5mm]approx.south east) {$\eps_X$};
\end{tikzpicture}
}

\def\figjmZ{
\begin{tikzpicture}[scale=0.75,transform shape,
	apparatus/.style={draw,fill=palegoldenrod,minimum height=18mm,minimum width=5mm},
	measurement/.style={draw,fill=blue!20}
	]
	\node[apparatus] (app) at (0,0) {\Large $\mathcal A_{X,Z}$};
	\draw[very thick] (app.west) -- node[below] {$S$} ++(-0.5,0);
	\draw[semitransparent] ([yshift=5.5mm]app.east) -- ++(0.6,0);
	\draw[semitransparent] ([yshift=4.5mm]app.east) -- node[below] {$R_X$} ++(0.6,0);
	\draw ([yshift=-4.5mm]app.east) -- ++(0.6,0);
	\draw ([yshift=-5.5mm]app.east) -- node[below] {$R_Z$} ++(0.6,0);
	
	\node[draw,fill=lightcoral,minimum size=12mm] (X) at (3.6,0) {\Large $\mathcal Q_Z$};
	\draw[very thick] (X.west) -- node[below] {$S$} ++(-0.5,0);
	\draw ([yshift=0.5mm]X.east) -- ++(0.6,0);
	\draw ([yshift=-0.5mm]X.east) -- node[below] {$R_Z$} ++(0.6,0);
	
	\node (approx) at (1.8,0) {\Large $\approx$};
	\node at ([xshift=0.5mm]approx.south east) {$\eps_Z$};
\end{tikzpicture}
}

\def\figjmall{
\begin{tikzpicture}
	\node (basic) at (0,0) {\figjmbasic};
	\node at (basic.north west) {a)};
	\node (jmX) at (4.20,1) {\figjmX};
	\node at (jmX.north west) {b)};
	\node (jmZ) at (4.20,-1) {\figjmZ};
	\node at (jmZ.north west) {c)};
\end{tikzpicture}}

\def\figidbasic{
\begin{tikzpicture}[
	apparatus/.style={draw=cornflowerblue!40!black,thick,fill=lightcfblue,minimum height=16mm,minimum width=10mm},
	measurement/.style={draw,fill=blue!20}
	]
	\node[apparatus] (app) at (0,0) {\Large $\mathcal A_X$};
	\draw[very thick] (app.west) -- node[below] {$S$} ++(-0.5,0);
	\draw ([yshift=5.5mm]app.east) -- ++(0.6,0);
	\draw ([yshift=4.5mm]app.east) -- node[below] {$R_X$} ++(0.6,0);
	\draw[very thick] ([yshift=-5mm]app.east) -- node[below] {$S'$} ++(0.6,0);
	\end{tikzpicture}
}

\def\figidXmeas{
\begin{tikzpicture}[scale=0.75,transform shape,
	apparatus/.style={draw=cornflowerblue!40!black,thick,fill=lightcfblue,minimum height=16mm,minimum width=10mm},
	measurement/.style={thick,draw=cornflowerblue!40!black,fill=lightcfblue,minimum size=10mm}
	]
	\node[apparatus] (app) at (0,0) {\Large $\mathcal A_X$};
	\draw[very thick] (app.west) -- node[below] {$S$} ++(-0.5,0);
	\draw ([yshift=5.5mm]app.east) -- ++(0.6,0);
	\draw ([yshift=4.5mm]app.east) -- node[below] {$R_X$} ++(0.6,0);
	\draw[very thick,semitransparent] ([yshift=-5mm]app.east) -- node[below] {$S'$} ++(0.6,0);
	
	\node at (1.7,0) (approx) {\Large $\approx$};
	\node at ([xshift=0.5mm]approx.south east) {$\eps_X$};
	
	\node[measurement] (ideal) at (3.4,0) {\Large $\mathcal Q_X$};
	\draw[very thick] (ideal.west) -- node[below] {$S$} ++(-0.5,0);
	\draw ([yshift=0.5mm]ideal.east) -- ++(0.6,0);
	\draw ([yshift=-0.5mm]ideal.east) -- node[below] {$R_X$} ++(0.6,0);
\end{tikzpicture}
}

\def\figidZmeas{
\begin{tikzpicture}[scale=0.75,transform shape,
	apparatus/.style={thick,draw=cornflowerblue!40!black,fill=lightcfblue,minimum height=16mm,minimum width=10mm},
	constant/.style={thick,draw=black!50,fill=lightgray!70,minimum height=16mm,minimum width=10mm},
	measurement/.style={draw=lightcoral!60!black,thick,fill=lightcoral!90,minimum height=10mm,minimum width=10mm}
	]
	\node[apparatus] (app) at (0,0) {\Large $\mathcal A_X$};
	\node[measurement] (pinch) at (-1.4,0) {\Large $\mathcal Q_Z^\natural$};
	\draw[very thick] (pinch) -- (app);
	\draw[very thick] (pinch.west) -- node[below] {$S$} ++(-0.5,0);
	\draw ([yshift=5.5mm]app.east) -- ++(0.6,0);
	\draw ([yshift=4.5mm]app.east) -- node[below] {$R_X$} ++(0.6,0);
	\draw[very thick] ([yshift=-5mm]app.east) -- node[below] {$S'$} ++(0.6,0);
	
	\node at (1.7,0) (approx) {\Large $\approx$};
	\node at ([xshift=0.5mm]approx.south east) {$\eta_Z$};
	
	\node[constant] (const) at (3.4,0) {\Large $\mathcal C$};
	\draw[very thick] (const.west) -- node[below] {$S$} ++(-0.5,0);
	\draw ([yshift=5.5mm]const.east) -- ++(0.6,0);
	\draw ([yshift=4.5mm]const.east) -- node[below] {$R_X$} ++(0.6,0);
	\draw[very thick] ([yshift=-5mm]const.east) -- node[below] {$S'$} ++(0.6,0);
\end{tikzpicture}
}

\def\figidall{
\begin{tikzpicture}
	\node (basic) at (0,0) {\figidbasic};
	\node at (basic.north west) {a)};
	\node (jmX) at (4.20,1) {\figidXmeas};
	\node at (jmX.north west) {b)};
	\node (jmZ) at (4.20,-1) {\figidZmeas};
	\node at (jmZ.north west) {c)};
\end{tikzpicture}
}

\begin{document}

\title{Operationally-Motivated Uncertainty Relations for Joint Measurability and the Error-Disturbance Tradeoff}
\author{Joseph M.\ Renes}
\author{Volkher B.\ Scholz}
\affiliation{Institute for Theoretical Physics, ETH Zurich, Wolfgang-Pauli-Str.\ 27, 8093 Zurich, Switzerland}

\begin{abstract}
	We derive new Heisenberg-type uncertainty relations for both joint measurability and the error-disturbance tradeoff for arbitrary observables of finite-dimensional systems. The relations are %are explicitly state-independent and are 
	formulated in terms of a directly operational quantity, namely the probability of distinguishing the actual operation of a device from its hypothetical ideal, by any possible testing procedure whatsoever. 
	Moreover, they may be directly applied in information processing settings, for example to infer that devices which can faithfully transmit information regarding one observable do not leak any information about conjugate observables to the environment. Though intuitively apparent from Heisenberg's original arguments, only more limited versions of this statement have previously been formalized.
\end{abstract}

\maketitle

\section{Introduction}
It is no overstatement to say that the uncertainty principle is a cornerstone of our understanding of quantum mechanics, %, insofar as we have any understanding beyond the ability to make use of the mathematical formalism. 
clearly marking the departure of quantum physics from the world of classical physics. 
Heisenberg's original formulation in 1927 mentions two facets to the principle. 
The first restricts the joint measurability of observables, stating that noncommuting observables can only be simultaneously determined with a characteristic amount of indeterminacy~\cite[p.\ 172]{heisenberg_uber_1927} (see \cite[p.\ 62]{wheeler_quantum_1984} for an English translation).
%\footnote{``...es wird gezeigt, da\ss{} kanonisch konjugierte Gr\"o\ss{}en simultan nur mit einer charakteristischen Ungenauigkeit bestimmt werden k\"onnen.''\cite[p.\ 172]{heisenberg_uber_1927}} 
The second describes an error-disturbance tradeoff, noting that the more precise a measurement of one observable is made, the greater the disturbance to noncommuting observables~\cite[p.\ 175]{heisenberg_uber_1927} (\cite[p.\ 64]{wheeler_quantum_1984}). %\footnote{``Im Augenblick der Ortsbestimmung...ver\"andert das Elektron seinen Impuls unstetig. Diese \"Anderung ist um so gr\"o\ss{}er...je genauer die Ortsbestimmung ist.''\cite[p.\ 175]{heisenberg_uber_1927}}
The two are of course closely related, and Heisenberg argues for the former on the basis of the latter. 

Precise formal statements corresponding to these two facets of the uncertainty principle were constructed only much later, 
due to the lack of precise mathematical descriptions of measurement and the measurement process in quantum mechanics. Here we must be careful to draw a distinction between statements addressing Heisenberg's original notions of uncertainty from those, like the standard Robertson uncertainty relation~\cite{robertson_uncertainty_1929}, which address the impossibility of finding a quantum state with well-defined values for noncommuting observables.   
Joint measurability has a longer history, going back at least to the seminal work of Arthurs and Kelly~\cite{arthurs_simultaneous_1965} and continuing in~\cite{she_simultaneous_1966,davies_quantum_1976,ali_systems_1977,prugovecki_fuzzy_1977,busch_indeterminacy_1985,busch_unsharp_1986,arthurs_quantum_1988,martens_towards_1991,ishikawa_uncertainty_1991,raymer_uncertainty_1994,leonhardt_uncertainty_1995,appleby_concept_1998,hall_prior_2004,werner_uncertainty_2004,ozawa_uncertainty_2004,watanabe_uncertainty_2011,busch_proof_2013,busch_heisenberg_2013,busch_measurement_2013}. %
Quantitative error-disturbance relations themselves have only been formulated relatively recently, going back at least to  Braginsky and Khalili~\cite[Chap.\ 5]{braginsky_quantum_1992} and continuing in~\cite{martens_disturbance_1992,appleby_concept_1998,ozawa_universally_2003,watanabe_quantum_2011,branciard_error-tradeoff_2013,buscemi_noise_2013,ipsen_error-disturbance_2013,coles_entropic_2013}. 

One motivation for finding formal uncertainty relations is to delineate the scope and validity of the uncertainty principle. Understanding the fundamental principles of any physical theory is an important endeavor, perhaps doubly so for quantum theory, whose conception of Nature differs so drastically from that used in classical mechanics, not to speak of everyday experience. 
Our motivation in this article is more operational, however. 
Here we are interested in finding Heisenberg-type uncertainty relations for joint measurability and error-disturbance that are useful for characterizing and analyzing (quantum) information processing tasks and are formulated in terms of quantities which are immediately relevant in such settings. 
Indeed, entropic uncertainty relations addressing state preparation~\cite{berta_uncertainty_2010,tomamichel_uncertainty_2011,berta_continuous_2013} have already been used to this end, in ensuring the security of quantum key distribution~\cite{tomamichel_tight_2012,furrer_continuous_2012}. 
The foundational and operational motivations are not completely distinct, as concrete information processing settings challenge us to find specific formalizations of the uncertainty principle.  

In this article we take a directly operational approach by quantifying error and disturbance in terms of the probability that the actual behavior of the apparatus can be distinguished from a relevant hypothetical behavior, in any experiment whatsoever.
We find new  uncertainty relations for both joint measurability (Figure~\ref{fig:jm}; Theorem~\ref{thm:jmvbs}) and the error-disturbance tradeoff (Figure~\ref{fig:ed}; Theorem~\ref{thm:ed}) of two arbitrary observables of discrete quantum systems. Our relations address the characteristics of measurement devices themselves, as opposed to entire experimental setups, and can be used in the analysis of quantum information processing tasks. Ultimately, all uncertainty relations spring from the same source, the requirement that the measurement process itself be treated as dynamical process according to the laws of quantum mechanics. The relations 
presented here are both relatively simple consequences of a basic structure theorem on quantum dynamics, the continuity of the Stinespring representation~\cite{kretschmann_continuity_2008,kretschmann_information-disturbance_2008}.

We quantify the error made by an apparatus in measuring an observable by the extent to which the apparatus can be distinguished from the \emph{ideal measurement} in  any possible experiment. 
Our uncertainty relation for joint measurability then relates the errors for each observable to a measure of the observables' incompatibility and implies both errors cannot simultaneously be small when the observables are incompatible. 
On the other hand, we quantify disturbance to an observable by how well the apparatus mimics one that produces a \emph{fixed output} when acting on states having well-defined values (eigenstates) of that observable. 
Our error-disturbance tradeoff then relates the error associated with measurement of one observable to the disturbance caused to the other and again implies that both cannot be small when the observables are incompatible.

As mentioned above, entropic state-preparation uncertainty relations haven proven useful in establishing the security of quantum key distribution. 
Our error-disturbance relation allows us to make a stronger statement, one useful in more general cryptographic scenarios beyond creation of secret keys: If a quantum system is subject to any kind of interaction with some external degrees of freedom which nevertheless still allows an experimenter to perform an approximately faithful measurement of a given observable, then the interaction approximately leaks no information about inputs which are eigenstates of a complementary observable. Thus, by appropriately examining the quantum state before and after the interaction, we can infer whether or not any information about the second observable has leaked to the external degrees of freedom. 
This notion is already present in Heisenberg's original arguments on an intuitive level, but directly operational versions have only been previously formalized for special cases. It can be used to construct leakage-resilient classical computers from fault-tolerant quantum computers~\cite{lacerda}. 

We have organized our results as follows. In the next section we define the distinguishability quantity and provide some background to the mathematical setting of the problem. We then present the joint-measurability and error-disturbance relations, whose proofs are deferred to the Methods section, Sec.~\ref{sec:methods}. In Sec.~\ref{sec:app} we discuss the applications to quantum information processing in more detail. Finally, we conclude in Sec.~\ref{sec:outlook} with a discussion of open questions raised by this work, in particular how our results could be extended to continuous-variable systems, and a comparison of our results with previous work. 

\section{Main Results}
\subsection{Background}
When working in the Schr\"odinger picture, any apparatus is described in the formalism of quantum theory by a completely positive, trace preserving operation, or \emph{quantum channel}~\cite{kraus_states_1983,nielsen_quantum_2000}. The channel $\cE$ maps states in the input state space, $\sS(\cH_A)$, to states in the output state space, $\sS(\cH_B)$. Here $\cH_A$ and $\cH_B$ are Hilbert spaces and $\sS(\cH)$ the set of bounded operators acting on $\cH$ and having unit trace. 

According to the Stinespring representation theorem~\cite{stinespring_positive_1955,paulsen_completely_2003}, any quantum channel $\cE:\sS(\cH_A)\to \sS(\cH_B)$ can be expressed in terms of an isometry $V:\cH_A\to \cH_B\otimes \cH_E$ involving an additional system $\cH_E$ as 
\begin{align}
\label{eq:stinespring}
	\cE(\rho)= \tr_E[V\rho V^\dagger],
\end{align}
for any $\rho\in\sS(\cH_A)$.
The extra system can be regarded as the \emph{environment} of the channel action, the additional degrees of freedom required to describe the dynamics of the pair by the Heisenberg (or Schr\"odinger) equation. The isometry is however not unique, but all possible Stinespring isometries are related by further isometries involving only the environmental degrees of freedom. 

Any particular Stinespring isometry naturally induces another channel, the complementary channel of $\cE$, denoted by $\cE^\sharp$, which maps $\sS(\cH_A)$ to $\sS(\cH_E)$ according to 
\begin{align}
	\label{eq:complementarychannel}
	\cE^\sharp(\rho)=\tr_B[V\rho V^\dagger].
\end{align} 
The main technical ingredient required for our results is the continuity of the Stinespring representation~\cite{kretschmann_continuity_2008,kretschmann_information-disturbance_2008}. This states that channels which are close (as measured by a particular norm) have Stinespring isometries which are also close. For the formal statement, see Theorem~\ref{thm:sc}.

The error and disturbance measures used here are 
formulated in terms of the probability $p_{\rm dist}(\cE,\cE')$ that one can distinguish the operation of one apparatus $\cE$ from another $\cE'$ in any test whatsoever, when the two are chosen with equal \emph{a priori} probability. Since this probability ranges from $\frac12$ (we can always just make a random guess) to $1$, it is more convenient to consider the distinguishability measure 
\begin{align}\delta(\cE,\cE'):=2p_{\rm dist}(\cE,\cE')-1,
\end{align}
which ranges from zero (completely indistinguishable) to one (completely distinguishable). 
Fortunately, the distinguishability is directly related to the norm used in the continuity of the Stinespring representation (see \eqref{eq:epsdefdn}). 
Nonetheless, the operational definition in terms of $p_{\rm dist}$ is sufficient to state our uncertainty relations, and we defer the more detailed presentation of $\delta(\cE,\cE')$ to Sec.~\ref{sec:methods}.

\subsection{Joint Measurability}
\begin{figure}[th]
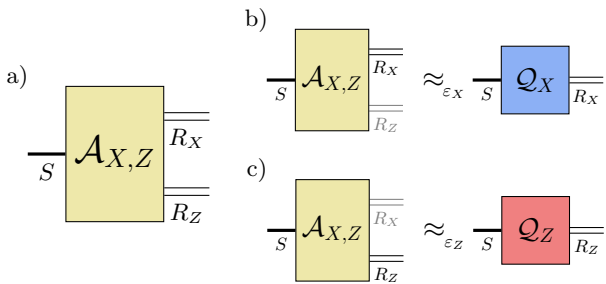

	\figjmall
	\caption{\label{fig:jm} a) An apparatus $\cA_{X,Z}$ designed to jointly measure two observables $X$ and $Z$ on a system $S$. It produces two results, the classical registers $R_X$ and $R_Z$. Ignoring either outcome amounts to nondeal measurement $\cM_X$ of $X$ and $\cM_Z$ of $Z$, shown in b) and c). The extent to which $\cM_X$ and $\cM_Z$ fail to simulate ideal measurement $\cQ_X$ of the observable $X$ is denoted by the error $\eps_X$ (defined in \eqref{eq:errordef}) and similarly $\eps_Z$ for the observable $Z$. The two errors are constrained by the joint measurability uncertainty relation \eqref{eq:jmur}.}  
\end{figure}
Let us now consider the question of joint measurability in more detail.  
As depicted in Figure~\ref{fig:jm},  joint measurability of two observables $X$ and $Z$ is naturally  concerned with how well a single apparatus $\cA_{X,Z}$ can simultaneously approximate both \emph{ideal measurements}, call them $\cQ_X$ and $\cQ_Z$. Any such device has of course two classical outputs, one for each observable, which we denote by $R_X$ and $R_Z$.  
The actual measurement $\mathcal M_X$ of $X$ made by the apparatus only takes the $R_X$ outcome into account, and similarly for $\cM_Z$. 
Then, we are specifically interested in the two types of \emph{error} inherent to  the apparatus, 
\begin{align}
\label{eq:errordef}
	\eps_X(\cA_{X,Z}) &:= \delta(\cM_X,\cQ_X) \quad \text{and}\\
	\eps_Z(\cA_{X,Z}) &:= \delta(\cM_Z,\cQ_Z).
\end{align}

We expect that, for incompatible or complementary observables, these quantities cannot both be small. In finite dimensions, we may quantify the complementarity of $X$ and $Z$ in terms of their eigenstates $\ket{\varphi_x}$ and $\ket{\vartheta_z}$, as follows. Letting $r(X;Z):=\tfrac1{\sqrt 2}\left(1-\min_x\max_z |\langle\varphi_x|\vartheta_z\rangle|^2\right)$, the  measure of complementarity is 
\begin{align}
	\label{eq:cjmdef}
	c_1(X,Z):=\max\{r(X;Z),r(Z;X)\}.
\end{align}
Then we have the following uncertainty relation,
\begin{theorem}[Joint Measurability]\label{thm:jmvbs}
	For any apparatus $\cA_{X,Z}$ which attempts to jointly measure two finite-dimensional observables $X$ and $Z$, 
\begin{align}
\label{eq:jmur}
	\eps_X(\cA_{X,Z})^{\half}+\eps_Z(\cA_{X,Z})^\half\geq c_1(X,Z).
\end{align}
\end{theorem}

The full proof is given in the Methods section, but we can sketch the main idea here. Since $\cM_X$ and $\cM_Z$ are defined from the same apparatus, they share a Stinespring isometry, say $V$. This isometry is close to appropriate isometries $W_X$ and $W_Z$ for $\cQ_X$ and $\cQ_Z$ as measured by $\eps_X$ and $\eps_Z$, respectively. By the triangle inequality for the isometry distance, we now have a relation for the distance between $W_X$ and $W_Z$,  which can be evaluated by making use of properties of the ideal measurements.

\subsection{Error-Disturbance Tradeoff}
\begin{figure}[th]
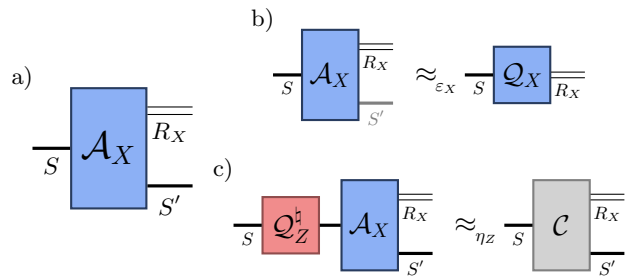

	\figidall 
	\caption{\label{fig:ed} a) An apparatus designed to extract information about the observable $X$ on a system $S$. The device produces two outputs: $R_X$, the classical register carrying information about $X$, and $S'$, the transformed quantum system. b) The error of the $X$ measurement is quantified by $\eps_X$, the extent to which $\cA_X$ approximates an ideal measurement $\cQ_X$. c) The disturbance $\eta_Z$ is quantified by how well the compound channel $\cA_X\circ \cQ_Z^\natural$ approximates a channel $\cC$ with a constant output (see~\eqref{eq:disturbdef}); here $\cQ_Z^\natural$ is an ideal non-selective measurement of $Z$. The error and disturbance are constrained by the uncertainty relation~\eqref{eq:ed}.}
\end{figure}
Next we turn to the tradeoff between the approximation error of a given apparatus $\cA_X$ for measuring observable $X$ and the disturbance caused to the observable $Z$. The setup is depicted in Figure~\ref{fig:ed}. Again $\cA_X$ produces the classical result in $R_X$, and the approximation error $\eps_X(\cA_X)$ is precisely the same as defined in the previous section.
Now we are also interested in the system $S'$ after the action of $\cA_X$, in particular the observable $Z$. 

One measure of disturbance to $Z$, natural in the Heisenberg picture where the apparatus changes observables on the system, not its state, is how closely the apparatus transforms 
$Z$ into (some multiple of) the identity operator. This way, any measurement of $Z$ after the action of the apparatus has nothing whatsoever to do with any properties of $Z$ which might have been present beforehand. 
But here we are after more: The disturbance to $Z$ should also hold  conditional on the measurement outcome in $R_X$. That is, it should not be possible to perform some subsequent ``recovery'' operation conditional on the measurement outcome which restores the $Z$ observable. This stronger notion of disturbance was used recently in~\cite{buscemi_noise_2013}.

To formulate a definition of disturbance that addresses this issue yet does explicitly include an optimization over recovery maps requires a little care, particularly in the Heisenberg picture. We first give the logic leading up to our definition for finite-dimensional systems in the Schr\"odinger picture, and then remark on a natural interpretation in the Heisenberg picture. We give both arguments, as the use of the Heisenberg picture is more convenient when considering infinite-dimensional systems.

 \begin{comment}
 Specifically, the maps $\cC'$ all have the form $\cC(A)=\tr[A\sigma]\id$ for some $\sigma\in \sS(\cH)$ and $A\in \sB(\cH)$, where $\sB(\cH)$ is the set of bounded operators.
 The Schr\"odinger picture equivalent of $\cC'$ is $\cC(\rho)=\sigma$ for all inputs $\rho$. We are thus led to a disturbance measure of the form
 \begin{align}
	\label{eq:disturbdef}
	\eta_Z(\mathcal A_X):=1-\inf_{\cC}\delta(\cA_X\circ \cQ_Z^\natural,\cC),
\end{align}
since better approximation means greater disturbance.  
\end{comment}

For finite-dimensional systems, disturbance to $Z$ in the Schr\"odinger picture amounts to its eigenstates all being mapped to a fixed output. In the worst case, this is true even when conditioning on the classical outcome of the $\cA_X$ apparatus. Therefore, our measure of disturbance is how well the action of $\cA_X$ approximates a channel with a constant output, when both are input with eigenstates of $Z$.  
To ensure that all inputs to $\cA_X$ are diagonal in the $Z$ basis, we may first perform the ideal non-selective measurement $\cQ_Z^\natural$, which measures the the state in the $Z$ basis and discards the result. The post-measurement state is necessarily diagonal in the $Z$ basis, and the map $\cQ_Z^\natural$ is a particular complement of the measurement $\cQ_Z$. Therefore, the disturbance is large if the map $\cA_X\circ \cQ_Z^\natural$ is close to a map $\cC$ which has constant output, say $\sigma$, for any input state $\rho$.  We are thus led to a disturbance measure of the form
 \begin{align}
	\label{eq:disturbdef}
	\eta_Z(\mathcal A_X):=1-\min_{\cC}\delta(\cA_X\circ \cQ_Z^\natural,\cC),
\end{align}
since a better approximation means greater disturbance. 

To motivate this definition in the Heisenberg picture, notice that the non-selective measurement $\cQ_Z^\natural$ has no effect on the $Z$ observable itself.  
Then, to the extent that $X$ and $Z$ are incompatible, $\cQ_Z^\natural$ followed by $\cA_X$ should completely scramble \emph{all} observables of the system. 
Indeed, this behavior is measured by \eqref{eq:disturbdef}, since the action of $\cC$ on observables is given by $\cC'(A)=\tr[A\sigma]\id$ for the same $\sigma\in \sS(\cH)$ and $A\in \sB(\cH)$, where $\sB(\cH)$ is the set of bounded operators.

As with joint measurement, we expect that both $\eps_X(\cA_X)$ and $\eta_Z(\cA_X)$ cannot both be small if $X$ and $Z$ are incompatible. For finite-dimensional observables we again measure complementarity in terms of the eigenvectors, but this time by the function 
\begin{align}
	c_2(X;Z):=1-\max_z\sum_x \{\tfrac1d-|\langle\varphi_x|\vartheta_z\rangle|^2\}_+,
\end{align}   
where $\{x\}_+=\max\{x,0\}$ and $d={\rm dim}(\cH_S)$. 
Then we have the following uncertainty relation, 
\begin{theorem}[Error-Disturbance Tradeoff]
	\label{thm:ed}
	For finite-dimensional observables $X$ and $Z$, any apparatus $\cA_X$ which attempts to gain information about observable $X$ satisfies 
	\begin{align}
		\label{eq:ed}
		\sqrt{2}\,\eps_X(\cA_X)^\half+\eta_Z(\cA_X)\geq c_2(X;Z).
	\end{align}
\end{theorem}
Again we give a brief sketch of the proof, which is detailed in the Methods section. The first step is to establish an intermediate result, which states any complement of a quantum channel which is close to a measurement $\cQ_X$ is itself close to the same measurement, possibly followed by preparation of a new quantum state conditioned on the measurement outcome. Then we consider the channel formed by preceding such a ``measure-prepare'' channel with the non-selective measurement in the $Z$ basis, $\cQ_Z^\natural$. Finally, the indistinguishability of the resulting joint channel from a constant-output channel turns on how close $\cQ_X\circ\cQ_Z^\natural$ is to a measurement with a fixed output distribution.  

The error-disturbance bound implies that when the error in $X$ is small, then the disturbance to $Z$ must be large relative to $c_2(X;Z)$. 
However, the opposite conclusion---low disturbance implies high error---does not follow from the bound, for two reasons. 
First, the disturbance quantity never quite reaches zero, since it is never possible to perfectly distinguish $\cA_X\circ \cQ_Z^\natural$ from a constant-output map $\cC$.
Second, even if $\eta_Z$ were zero and the observables conjugate so that $c_2(X;Z)=1$, $\eps_X$ would still only necessarily be at least $\frac 12$.

\section{Applications in Quantum Information Processing}
\label{sec:app}
A useful tool in the construction of quantum information processing protocols is the link between reliable transmission of $X$ eigenstates through a channel $\cN$ and $Z$ eigenstates through its complement $\cN^\sharp$, particularly when the observables $X$ and $Z$ are maximally complementary, i.e.\ $|\langle \varphi_x|\vartheta_z\rangle|^2=\frac 1d$ for all $x,z$. Due to the uncertainty principle, we expect that a channel cannot reliably transmit the bases to different outputs, since this would provide a means to simultaneously measure $X$ and $Z$. 
This link has been used by Shor and Preskill to prove the security of quantum key distribution~\cite{shor_simple_2000} and by Devetak to determine the quantum channel capacity~\cite{devetak_private_2005}.
Entropic state-preparation uncertainty relations from \cite{berta_uncertainty_2010,tomamichel_uncertainty_2011} can be used to understand both results, as shown in~\cite{renes_duality_2011,renes_physics_2012}. 

However, the above approach has the serious drawback that it can only be used in cases where the specific $X$-basis transmission over $\cN$ and $Z$-basis transmission over $\cN^\sharp$ are in some sense compatible and not \emph{counterfactual}; because the argument relies on a state-dependent uncertainty principle, both scenarios must be compatible with the same quantum state. 
Fortunately, this can be done for both QKD security and quantum capacity, because at issue is whether $X$-basis ($Z$-basis) transmission is reliable (unreliable) on average \emph{when the states are selected uniformly at random}. Choosing among either basis states at random is compatible with a random measurement in either basis of half of a maximally-entangled state, and so both $X$ and $Z$ basis scenarios are indeed compatible. The same restriction to choosing input states uniformly appears in the recent result of \cite{buscemi_noise_2013}, as it also ultimately relies on a state-preparation uncertainty relation.

Using Theorem~\ref{thm:ed} we can extend the method above to counterfactual uses of arbitrary channels $\cN$, in the following sense: If acting with the channel $\cN$ does not substantially affect the possibility of performing an $X$ measurement, then $Z$-basis inputs to $\cN^\sharp$ result in an essentially constant output. More concretely, we have %the following 
\begin{corollary}
\label{cor:leakage}
	Given a channel $\cN$ and complementary channel $\cN^\sharp$, suppose that there exists a measurement $\Lambda_X$ such that $\delta(\cQ_X,\Lambda_X\circ \cN)\leq \eps$. Then there exists a constant channel $\cC$ such that $\delta(\cN^\sharp\circ \cQ_Z^{\natural},\cC)\leq 2\sqrt\eps+1-c_2(X;Z)$. For maximally complementary $X$ and $Z$, $\delta(\cN^\sharp\circ \cQ_Z^{\natural},\cC)\leq 2\sqrt\eps$.
\end{corollary}
\begin{proof}
Let $V_\cN$ be the Stinespring dilation of $\cN$ such that $\cN^\sharp$ is the complementary channel and define $\cA_X=\Lambda_X\circ V_\cN$. For $\cC$ the optimal choice in the definition of $\eta_Z(\cA_X)$, \eqref{eq:ed} implies $\delta(\cA_X\circ \cQ_Z^\natural,\cC)\leq 2\sqrt \eps+1-c_2(X;Z)$. Since $\cN^\sharp$ is obtained from $\cA_X$ by ignoring the $\Lambda_X$ measurement result, $\delta(\cN^\sharp\circ \cQ_Z^\natural,\cC)\leq \delta(\cA_X\circ \cQ_Z^\natural,\cC)$.%, completing the proof.
\end{proof}

This formulation is important because in more general cryptographic and communication scenarios we are interested in the worst-case behavior of the protocol, not the average case under some particular probability distribution. For instance, in~\cite{lacerda} the goal is to construct a classical computer resilient to leakage of $Z$-basis information by establishing that reliable $X$ basis measurement is possible despite the interference of the eavesdropper. However, such an $X$ measurement is entirely counterfactual and cannot be reconciled with the actual $Z$-basis usage, as the $Z$-basis states will be chosen \emph{deterministically} in the classical computer.

\section{Outlook}
\label{sec:outlook}

A number of open questions present themselves, beyond an extension of our results to continuous variable systems, which is discussed at length below. 
First, it would be interesting to examine the optimality conditions of the semidefinite program used in the proofs of both relations to see if the bounds presented here could be improved. Both are somewhat weak in certain extreme cases: In the joint-measurability relation~\eqref{eq:jmur} even if one measurement is perfect the bound on the error of the other is only at least one half, while from the error-disturbance relation~\ref{eq:ed} one cannot conclude that low disturbance implies high error. One could also examine the tightness of either relation (at least numerically) for specific measurement devices, for instance the measurement used in the experimental tests~\cite{erhart_experimental_2012,rozema_violation_2012} of Ozawa's error-disturbance relation~\cite{ozawa_universally_2003}. Finally, one could also derive bounds on joint-measurability from the error-disturbance tradeoff itself and and see how it compares with~\eqref{eq:jmur}.

\subsection{Extension to Continuous Variables}
\label{sec:cv}

Starting from Heisenberg's seminal paper, clearly the most well-studied uncertainty relations involve the conjugate pair position and momentum. 
Hence it is desirable to extend our results to this setting as well. The two main technical tools used establish Theorems~\ref{thm:jmvbs} and~\ref{thm:ed} are the continuity of the Stinespring representation for completely positive maps (Theorem~\ref{thm:sc}, as well as formulation of the completely bounded norm as a semidefinite program involving the Choi representative of a map (Equations~\ref{eq:primal} and \ref{eq:dual}). Both results can be carried over to the case of infinite dimensional systems. The continuity theorem remains valid as stated in Theorem~\ref{thm:sc}, see~\cite{kretschmann_continuity_2008}. In addition, since we are concerned with the distinguishability of measurements, which are channels that destroy all entanglement (known as entanglement-breaking channels), the concept of Choi representatives carries over as well~\cite[Theorem 3]{owari_squeezing_2008}. 

Before applying these two tools in the infinite-dimensional setting, we must first ensure that the definitions of error and disturbance are sensible. Recall that both rest on the probability of distinguishing between two channels, maximized over all input states and observables. In infinite dimensions this includes observables of arbitrary \emph{precision}; any operationally valid distinguishability measure should however also take precision into account. For finite dimensions precision is not an issue due to the inherent discreteness of the results (infinite precision is not unphysical, in some sense).  

For the case of our error-disturbance relation, the extension to continuous variables is even more mathematically delicate, as there simply is no instrument which implements an ideal position or momentum measurement, including the post-measurement quantum state~\cite[Theorem 3.3]{davies_quantum_1976}. Instruments implementing imprecise measurements do however exist. A meaningful extension of our joint measurability result would not need to address this difficulty, since ideal position and momentum measurements themselves do exist, ignoring the post-measurement quantum state.

In the case of distinguishing measurements, as here, the observables in question are so-called test functions, which take values between 0 and 1, generalizing characteristic or indicator functions associated with a given subset of the measurement output space (the real line for position or momentum). 
One way to limit the precision of test functions would be to ensure they are ``smeared out'' over some mimimum length scale, for instance by composing all test functions with a physical noise channel. Our proof technique could potentially be adapted to include this additional step, since by using a physical noise channel one could still make use of the Stinespring representation. 

Another option would be to restrict to slowly-varying functions, that is having a bounded Lipschitz constant, which are thereby  insensitive to changes on small length scales. This approach is chosen in the work of Werner~\cite{werner_uncertainty_2004}, and the resulting measure of distinguishability between probability distributions is the Wasserstein metric of order one. Thanks to the theory of optimal transport, there is a nice dual interpretation of the distinguishability measure: It is the cost required to change one distribution into the other, as measured by the distance on the real line. Note that the variational distance, which underlies  the completely bounded norm as used here, can be formulated similarly. Now the cost is measured by a different metric which is simply zero if the two values are identical and one otherwise~\cite{villani2008optimal}. Hence a possible way to incorporate finite precision limits into our setup would be to require the metric to be only sensitive to differences above some finite minimal length scale. It would be interesting to formalize both approaches and in particular their ``completely bounded'' versions.

\subsection{Comparison to Previous Results}
\label{sec:comparison}

In the recent work of Busch, Lahti and Werner~\cite{busch_proof_2013}, the authors used the Wasserstein metric of order two, corresponding to the mean squared error, as the underlying distance $D(.,.)$ to measure the closeness of probability distributions. If $\cM^Q$, $\cM^P$ are the marginals of some joint measurement of position $Q$ and momentum $P$, and $X_\rho$ denotes the distribution coming from applying the measurement $X$ to the state $\rho$, their relation reads
\begin{align}
	\label{eq:relation_wernergang}
	\sup_\rho D(\cM^Q_\rho,Q_\rho) \cdot \sup_\rho D(\cM^P_\rho,P_\rho) \geq c \,,
\end{align}
for some universal constant $c$. In~\cite{busch_measurement_2013}, the authors generalize their results to arbitrary Wasserstein metrics. As in our case, the two distinguishability quantities in \eqref{eq:relation_wernergang} are separately maximized over all states, and hence the resulting expression characterizes the goodness of the approximate measurement. 

One could instead ask for a ``coupled optimization'', a relation of the form 
\begin{align}
\label{eq:coupledsup}
\sup_\rho \left[D(\cM^Q_\rho,Q_\rho) D(\cM^P_\rho,P_\rho) \right]\geq c',
\end{align}
for some other constant $c'$.\footnote{Such an approach has been advocated by David Reeb (private communication).} While this statement certainly tells us that no device can accurately measure both position and momentum for all input states, the bound $c'$ only holds (and can only hold) for the worst possible input state. In contrast, the bounds found in Theorems~\ref{thm:jmvbs} and \ref{thm:ed}, as well as in \eqref{eq:relation_wernergang} are {state-independent} in the sense that the bound holds for all states. Indeed, the two approaches are more distinct than the similarities between \eqref{eq:relation_wernergang} and \eqref{eq:coupledsup} would suggest. By optimizing over input states separately, our results and those of \cite{werner_uncertainty_2004,busch_proof_2013,busch_measurement_2013} are statements about the properties of measurement devices themselves, independent of any particular experimental setup. State-dependent settings capture the behavior of measurement devices in specific experimental setups and must therefore account for the details of the input state. 

The same set of authors also studied the case of finite-dimensional systems, in particular qubit systems, again using the Wasserstein metric of order two~\cite{busch_heisenberg_2013}. Their results for this case are similar, with the product in \eqref{eq:relation_wernergang} replaced by a sum. Perhaps most closely related to our results is the recent work by Ipsen~\cite{ipsen_error-disturbance_2013}, who uses the variational distance as the underlying distinguishability measure to derive similar additive uncertainty relations. We note, however, that both \cite{busch_heisenberg_2013} and \cite{ipsen_error-disturbance_2013} only consider joint measurability and do not consider the change to the state after the approximate measurement is performed, as it is done in our error-disturbance relation. Furthermore, both base their distinguishability measures on the measurement statistics of the devices alone. But this does not necessarily tell us how distinguishable two devices ultimately are, as we could employ input states entangled with ancilla systems to test them. These two measures can be different~\cite{kitaev_quantum_1997}, even for entanglement-breaking channels~\cite{sacchi_entanglement_2005}.  

Entropic quantities are another means of comparing two probability distributions, an approach taken recently by Buscemi \emph{et al.}~\cite{buscemi_noise_2013} and Coles and Furrer~\cite{coles_entropic_2013} (see also Martens and de~Muynck~\cite{martens_disturbance_1992}). Both contributions formalize error and disturbance in terms of relative or conditional entropies, and derive their results from entropic uncertainty relations for state preparation which incorporate the effects of quantum entanglement~\cite{berta_uncertainty_2010,tomamichel_uncertainty_2011}. They differ in the choice of the entropic measure and the choice of the state on which the entropic terms are evaluated. Buscemi \emph{et al.} find state-independent error-disturbance relations involving the von Neumann entropy, evaluated for input states which describe  observable eigenstates chosen uniformly at random. As described in Sec.~\ref{sec:app}, the restriction to uniformly-random inputs is significant, and leads to a characterization of the average-case behavior of the device (averaged over the choice of input state), not the worst-case behavior as presented here. 
Meanwhile, Coles and Furrer make use of general R\'enyi-type entropies, hence also capturing the worst-case behavior. However, they are after a state-dependent error-disturbance relation which relates the amount of information a measurement device can extract from a state about the results of a \emph{future} measurement of one observable to the amount of disturbance caused to other observable. 

An important distinction between both these results and those presented here is the quantity appearing in the uncertainty bound, i.e.\ the quantification of complementarity of two observables. As both the aforementioned results are based on entropic state-preparation uncertainty relations, they each quantify complementarity by the largest overlap of the eigenstates of the two observables. This bound is trivial should the two observables share an eigenstate. However, a perfect joint measurement is clearly impossible even if the observables share all but two eigenvectors (if they share all but one, they necessarily share all eigenvectors). Both $c_1(X,Z)$ and $c_2(X;Z)$ used here are nontrivial whenever not all eigenvectors are shared between the observables.

%\section{Further Open Questions}

\acknowledgements{Thanks to Omar Fawzi, Fabian Furrer, David Reeb, and Michael Walter for helpful discussions.
This work was supported by the by the German Science Foundation (grant CH 843/2-1), the Swiss National Science Foundation (through the National Centre of Competence in Research `Quantum Science and Technology' and grants No. 200020-135048, PP00P2-128455, 20CH21-138799 (CHIST-ERA project CQC)), by the European Research Council (grants No. 258932 and No. 337603) as well as by the Swiss State Secretariat for Education and Research through COST action MP1006. VBS is supported by an ETH Postdoctoral Fellowship.}

\section{Methods}
\label{sec:methods}
\subsection{Mathematical Setup}
A result of Helstrom~\cite{helstrom_detection_1967,helstrom_quantum_1976} shows that the distinguishability of two quantum \emph{states} $\rho_1$ and $\rho_2$ is precisely their trace distance, 
 $\delta(\rho_1,\rho_2):=\tfrac12\onenorm{\rho_1-\rho_2}$, where $\onenorm{A}:=\tr[\sqrt{A^\dagger A}]$. In other words, $p=\half(1+\delta(\rho_1,\rho_2))$. Distinguishability of states can be transferred to that of channels by asking for the most distinguishable states that two channels could produce from a common input. Since this distinguishability can be enhanced for inputs which are entangled with ancillary systems unaffected by the channel itself, one is lead to consider the \emph{diamond norm} of quantum channels~\cite{kitaev_quantum_1997}. 
 
 For a channel $\cE:\sB(\cH_A)\to \sB(\cH_B)$, the diamond norm is defined by 
\begin{align}
	\dnorm{\cE}:=\sup_{k\geq 1}\sup_{ \rho\in \sS(\cH_A\otimes \mathbb{C}^k)}\onenorm{\cE\otimes \cI_k(\rho)},
\end{align}
where $\cI_k$ is the identity channel from $\sB(\mathbb{C}^k)$ to itself which just reproduces its input. 
Using the diamond norm we arrive at the following distinguishability measure for quantum channels, 
\begin{align}
	\label{eq:epsdefdn}
	\delta(\cE_1,\cE_2)=\tfrac12\dnorm{\cE_1-\cE_2}.
\end{align}

\begin{comment}
Back in the Heisenberg picture, the 
the equivalent norm is the dual of the diamond norm,  the completely-bounded operator norm~\cite{paulsen_completely_2003}. The operator norm itself is defined by $\opnorm{A}=\sup_{\psi\in \cH} \bra\psi A\ket\psi$ and the completely-bounded version for channels is given by 
\begin{align}
	\cbnorm{\cE}:=\sup_{k\geq 1}\sup_{A\in \sB(\cH_B\otimes \mathbb C^k)}\opnorm{\cE\otimes \cI_k(A)}.
\end{align} 
Therefore, the $\delta$ also takes the form
\begin{align}
	\label{eq:epsdefcb}
	\delta(\cE_1,\cE_2)=\tfrac12 \cbnorm{\cE_1-\cE_2}.
\end{align}
\end{comment}

This expression is not closed-form, as an optimization is required to evaluate the diamond  norm. 
However, in finite dimensions the diamond norm can be cast as a convex optimization, specifically as a semidefinite program~\cite{watrous_semidefinite_2009}. This makes numerical approximation tractable and will be analytically useful in the proofs to follow. 

Given a Hilbert space $\cH$ with basis $\{\ket{k}\}_{k=1}^d$, define $\ket\Omega=\sum_{k=1}^d\ket{k}\otimes \ket{k}\in \cH\otimes\cH$. Then, for any channel  $\cE:\sB(\cH_A)\to \sB(\cH_B)$, let $\sC$ 
denote the Choi mapping of $\cE$ to a bipartite state,
\begin{align}
	\sC(\cE):=\cE\otimes\cI(\ketbra\Omega)\in \sS(\cH_B\otimes \cH_A).
\end{align}
The diamond norm can then be expressed as
\begin{align}\label{eq:primal}
	\tfrac12\dnorm{\cE} =& \max \tr[\sC(\cE)Y]\\
	& \,{\rm s.t.}\, Y\leq \id_B\otimes \rho_A,\nonumber\\
	& \,\phantom{\rm s.t.}\, Y\geq 0,\nonumber\\
	& \,\phantom{\rm s.t.}\, \rho\in \sS(\cH_A)\nonumber.
\end{align}
This semidefinite program also comes in the dual form
\begin{align}\label{eq:dual}
	\tfrac12\dnorm{\cE} =& \min \opnorm{\tr_B(R)}\\
	& \,{\rm s.t.}\, R\geq \sC(\cE),\nonumber\\
	&\,\phantom{\rm s.t.}\, R\geq 0\nonumber
\end{align}
Here we have used the operator norm, defined by $\opnorm{A}=\sup_{\psi\in \cH} \bra\psi A\ket\psi$.
Both forms of the semidefinite program will be useful in the proofs to follow.

Finally, we arrive at the central technical tool required for our uncertainty relations, the continuity of the Stinespring representation for finite-dimensional systems.
\begin{theorem}[Stinespring Continuity \cite{kretschmann_information-disturbance_2008}]\label{thm:sc}
	Given two quantum channels $\mathcal E_1,\mathcal E_2:\sB(\cH_A)\to \sB(\cH_B)$ with corresponding Stinespring isometries $V_1:\cH_A\to \cH_B\otimes \cH_{E_1}$ and $V_2:\cH_A\to \cH_B\otimes \cH_{E_2}$, we have
	\begin{align*}
		\min_{U}\opnorm{UV_1-V_2}^2\leq \dnorm{\cE_1-\cE_2}\leq 2\min_{U}\opnorm{UV_1-V_2},
	\end{align*}
	with the minimum taken over all isometries $U:\cH_{E_1}\to \cH_{E_2}$.
	\end{theorem}

\subsection{Joint Measurability}
A device $\cA_{X,Z}$ jointly measuring two observables $X$, $Z$ on a Hilbert space $\cH_S$ can be modeled by an isometry 
\begin{align}
	\begin{split}
	V&:\cH_S \to \cH_{R} \otimes \cH_{\hat X} \otimes \cH_{\hat X'}\otimes \cH_{\hat Z}\otimes \cH_{\hat Z'} \\
	V& = \sum_{xz} M_{xz} \otimes \ket{x}_{\hat X} \otimes \ket{x}_{\hat X'} \otimes \ket{z}_{\hat Z} \otimes \ket{z}_{\hat Z'}.
	\end{split}
\end{align}
Here, the Hilbert spaces $\cH_{\hat X}$, $\cH_{\hat Z}$ serve to record the (classical) measurement result, and doubling them ensures that no quantum coherence is present in systems $\hat X$ or $\hat Z$ alone. The spaces have dimension equal to the number of outputs, and $x$, $z$ label arbitrary bases (quite possibly the same). The operators $M_{x,z}:\cH_S\to \cH_R$ are the Kraus operators of $\cA_{X,Z}$.
 The maps $\cM_X:\sS(\cH_S)\to \sS(\cH_{\hat X})$ and $\cM_Z:\sS(\cH_S)\to \sS(\cH_{\hat Z})$ describing the measurement of $X$ or $Z$ by the device are determined by 
\begin{align}
	\cM_X(\rho) &= \tr_{R\hat Z\hat Z'\hat X'}[V \rho V^\dagger], \quad\text{and}\\
	\cM_Z(\rho) &=\tr_{R\hat X\hat X'\hat Z'}[V \rho V^\dagger].
\end{align} 
Particular Stinespring isometries for the ideal measurements $\cQ_X$, $\cQ_Z$ are given by 
	\begin{align} 
	\begin{split}		W_X& : \cH_S \to \cH_S \otimes \cH_{\hat X}\otimes \cH_{\hat X'},\\	W_X& = \sum_x Q_X(x) \otimes \ket{x}_{\hat X} \otimes \ket{x}_{\hat X'}, \quad \text{and}
	\end{split}\label{eq:QXisometry}\\
	\begin{split}
		W_Z &: \cH_S \to \cH_S \otimes \cH_{\hat Z}\otimes \cH_{\hat Z'},\\	W_Z &= \sum_z Q_Z(z) \otimes \ket{z}_{\hat Z}\otimes \ket{z}_{\hat Z'},
		\end{split}
	\end{align}
where $Q_X(x)$ and $Q_Z(z)$ denote the projection operators associated with the observables $X$ or $Z$. In terms of the eigenvectors of the respective observables, $Q_X(x)=\ketbra{\varphi_x}$ and $Q_Z(z)=\ketbra{\vartheta_z}$.
	
According to the lower bound in the Stinespring continuity Theorem \ref{thm:sc}, there exist isometries $U_X:\cH_S \to \cH_R \otimes \cH_{\hat X} \otimes \cH_{\hat X'}$ and $U_Z: \cH_S \to \cH_R \otimes \cH_{\hat Z} \otimes \cH_{\hat Z'}$ such that 	\begin{align}
		\opnorm{V - U_X\,W_Z}^2 &\leq \dnorm{\cM_Z - \cQ_Z} ,\\
		\opnorm{V - U_Z\,W_X}^2 &\leq \dnorm{\cM_X - \cQ_X} .
	\end{align} 
	Using the triangle inequality and the definition of error $\eps$, we find  
	\begin{align}
		\tfrac1{\sqrt{2}}\opnorm{U_XW_Z-U_ZW_X}&\leq \eps_X(\cA_{X,Z})^\half+\eps_Z(\cA_{X,Z})^\half. 
	\end{align}
	
In principle, this inequality already gives a bound on the errors $\eps_X$ and $\eps_Z$. 
However, it is implicitly a function of the measurement device, since $U_X$ and $U_Z$ are only characterized by the optimal choice in the Stinespring dilation, which itself turns on the description of the device. 
For finite-dimensional systems, we can find a bound which holds for all devices as follows.  
	
	\begin{proof}[Proof of Theorem~\ref{thm:jmvbs}]
First define the map $\cE_Z: \sB(\cH_S)\to \sB(\cH_{\hat X})$ by 
\begin{align}
	\cE_Z(\rho)&:=\tr_{R\hat Z\hat Z'\hat X'}[U_XW_Z\rho W_Z^\dagger U_X^\dagger],\\
	&\phantom{:}=\tr_{R\hat X'}[U_X\cQ_Z^\natural(\rho)U_X^\dagger].
\end{align}
Here, we have used the map $\cQ_Z^\natural$, which will be used later in the definition of disturbance; it is sometimes called a ``pinch map'' and is defined by 
\begin{align} 
	\cQ_Z^\natural(\rho) &= \sum_z Q_Z(z)\rho\, Q_Z(z)\\
	&= \sum_z \bra{\vartheta_z}\rho\ket{\vartheta_z}\, \ketbra{\vartheta_z}_S.
\end{align}

The ideal measurement  is $\cQ_X(\rho)=\tr_{S\hat X'}[W_X\rho W_X^\dagger]$, which can also be expressed as $\cQ_X(\rho)=\tr_{R\hat Z\hat Z'\hat X'}[U_ZW_X\rho W_X^\dagger U_Z^\dagger]$. By Stinespring continuity we therefore have  
\begin{align}
	 \opnorm{U_XW_Z - U_ZW_X} \geq \tfrac 12 \dnorm{\cQ_X-\cE_Z}.
	\end{align}
	
	Now we make use of the primal form of the semidefinite program given in~\eqref{eq:primal}. We are free to choose the basis in which $\ket\Omega$ is defined, so let us select $\{\ket{\varphi_x}\}$, the eigenbasis of the $X$ observable. Then we find $\sC(\cQ_X)=\sum_x \ketbra x_{\hat X}\otimes \ketbra{\varphi_x}_{S}$, while
\begin{align}
	\sC(\cE_Z)=&\sum_{z,y,y'}\bra{\varphi_{y'}}Q_Z(z)\ket{\varphi_y}\,\nonumber\\
	&\;\times \tr_{R\hat X'}[U_XQ_Z(z)U_X^\dagger]\otimes \ket{\varphi_y}\bra{\varphi_{y'}}_S.
\end{align}
Now define $\Lambda_x:= U_X^\dagger(\id_R\otimes \ketbra x_{\hat X}\otimes\id_{\hat X'}) U_X$ and let $Y=\ketbra x_{\hat X}\otimes \ketbra{\varphi_{x}}$ for some $x$. From~\eqref{eq:primal} we get
\begin{align}
	\tfrac 12& \dnorm{\cQ_X-\cE_Z} \nonumber\\&\geq \max_{x} \big(1-\sum_z \tr[Q_Z(z)\Lambda_x]\, |\langle\vartheta_z|\varphi_x\rangle|^2\big)\\
	&\geq 1-\min_x\big(\max_z |\langle\vartheta_z|\varphi_x\rangle|^2\sum_z\tr[Q_Z(z)\Lambda_x]\big)\\
	&=1-\min_x\big(\max_z |\langle\vartheta_z|\varphi_x\rangle|^2\,\tr[\Lambda_x]\big)\\
	&\geq1-\min_x\big(\max_z |\langle\vartheta_z|\varphi_x\rangle|^2\big)\,\min_x\tr[\Lambda_x]\\
	&\geq 1-\min_x\max_z |\langle\vartheta_z|\varphi_x\rangle|^2.
\end{align}
The last inequality follows because $\sum_x \Lambda_x=U_X^\dagger U_X=\id_S$, which then implies $\sum_{x=1}^d\tr[\Lambda_x]=d$ and therefore $\min_x \tr[\Lambda_x]\leq 1$.
\end{proof} 
	
\begin{comment}
	For joint measurements of position and momentum, we restrict attention to covariant observables. Making use of their structure, we can again find a bound which holds for all covariant joint observables, as considered in~\cite{werner_uncertainty_2004}.
	\begin{proof}[Proof of Thereom~\ref{thm:jmvbs} for covariant observables]
		g
	\end{proof}
	\end{comment}
		
	\subsection{Error-Disturbance Tradeoff}
	A device $\cA_X$ that measures $X$ on a Hilbert space $\cH_S$ and also produces an output state on $\cH_R$ can be modelled by an isometry
	\begin{align}
		\begin{split}
		V&:\cH_S\to \cH_R\otimes \cH_E\otimes \cH_{\hat X} \otimes \cH_{\hat X'} \\
	V& = \sum_{x} M_{x} \otimes \ket{x}_{\hat X} \otimes \ket{x}_{\hat X'}.
	\end{split}\label{eq:AXisometry}
\end{align}
Again the $M_x:\cH_S\to \cH_R\otimes \cH_E$ are the Kraus operators of the map, while $\cH_E$ is an extra system which may be needed to purify the output in $R$. As with $\cA_{X,Z}$, the measurement outcome is recorded in $\cH_{\hat X}$ (and $\cH_{\hat X'}$). 

Before proceeding to the proof of Theorem~\ref{thm:ed}, we first establish an intermediate result which states  that complementary channels of approximate measurements are themselves approximate measurements, possibly followed by state preparation. 
More precisely, if $\cA_X$ approximates the ideal measurement of $X$, then the output in system $R$ can be simulated by a map which simply prepares a state on $\cH_R$ conditional on the result $x$ of the measurement. 
We call $\cP:\sS(\cH_{\hat X})\to \sS(\cH_R\otimes \cH_{\hat X})$ a conditional preparation channel if it has the action $\cP(\ketbra x)=\rho^x_R\otimes \ketbra x_{\hat X}$, for some states $\rho^x_R$. Then, by Stinespring continuity we can show  	
\begin{theorem}\label{thm:measprep}
	Given a channel $\cA_X:\sS(\cH_S)\to \sS(\cH_R\otimes \cH_{\hat X})$, let $\cM_X:\sS(\cH_S)\to \sS(\cH_{\hat X})$ be just the output in $\hat X$, i.e.\ $\cM_X(\rho)=\tr_{R\hat X'}[\cA_X(\rho)]$. If 
	\begin{align}
		\dnorm{\cM_X-\cQ_X}\leq \eps,
	\end{align}
then there exists a conditional state preparation channel $\cP:\sS(\cH_{\hat X})\to \sS(\cH_R\otimes \cH_{\hat X})$ such that 
	\begin{align}
		\dnorm{\cA_X-\cP\circ \cQ_X}\leq 2\sqrt\eps.
	\end{align}
\end{theorem}
\begin{proof}
	We can reuse the Stinespring dilation $W_X$ of $\cQ_X$ given in \eqref{eq:QXisometry}, while the dilation of $\cA_X$ is given by \eqref{eq:AXisometry}. By the lower bound in Theorem~\ref{thm:sc}, the premise above implies 
	\begin{align}
		\opnorm{V-UW_X}\leq \sqrt\eps,
	\end{align}
for some isometry $U:\cH_S\to \cH_R\otimes \cH_E$. 

Now, for any $\rho\in \sS(\cH_S)$, 
\begin{align}	
\cA_X(\rho)=\tr_{E\hat X'}[V\rho V^\dagger].
\end{align} 
Thus, defining the map $\cE:\sS(\cH_S)\to \sS(\cH_R\otimes \cH_{\hat X})$ by 
\begin{align}
\cE(\rho):&=\tr_{E\hat X'}[UW_X\rho W_X^\dagger U^\dagger] \\
&=\tr_{E}[U\tr_{\hat X'}[W_X\rho W_X^\dagger] U^\dagger].
\end{align}
the upper bound in Theorem~\ref{thm:sc} implies
\begin{align}
	\dnorm{\cA_X-\cE}\leq 2\sqrt\eps.
\end{align}

All that remains to show is that $\cE=\cP\circ \cQ_X$ for some conditional preparation channel $\cP$. Using the form of $W_X$ we can express the action of $\cE$ as 
\begin{align}
	\cE(\rho) &= \sum_x\tr_{E}[U Q_X(x)\rho Q_X(x) U^\dagger]\otimes \ketbra x_{\hat X}\\
	&= \sum_x \bra{\varphi_x}\rho\ket{\varphi_x}\,\tr_E[U\ketbra{\varphi_x}U^\dagger]\otimes \ketbra x_{\hat X}.
\end{align}
This is a conditional preparation channel $\cP$ for $\rho^x_R=\tr_E[U\ketbra{\varphi_x}U^\dagger]$. 
\end{proof}

\begin{comment}
\begin{lemma}
	Let $\cP$ be a conditional preparation channel $\cP:L^\infty(X,\sB(\cH_T))\to L^\infty(X)$ with the action $\cP:(B_x)_x\mapsto f(x)=\tr[B_x\rho_x]$ for some normalized states $\rho_x$ and $\cC:L^\infty(X,\sB(\cH_T))\to \sB(\cH_S)$ a constant channel with the action $\cC:(B_x)_x\mapsto \int_X\!{\rm d}\mu_X(x)\,\tr[B_x\sigma_x]\id$ for some state $(\sigma_x)_x\in L^1(X,\sB(\cH_T))$.
	\end{lemma}
	\end{comment}

Now we can establish the error-disturbance bound. 
\begin{proof}[Proof of Theorem~\ref{thm:ed}]
An apparatus $\cA_X$ with error $\eps_X(\cA_X)$ satisfies the premise of Theorem~\ref{thm:measprep} with $\eps=2\eps_X(\cA_X)$, and therefore 
 	\begin{align}
 		\delta(\cA_X,\cP\circ \cQ_X)\leq \sqrt{2\eps_X(\cA_X)}
 	\end{align} 
for some conditional preparation channel $\cP$. 
By the triangle inequality, for any map $\cC$ we have
\begin{align}
\delta(\cA_X\circ \cQ_Z^\natural,\cC)\leq&\, \delta(\cA_X\circ \cQ_Z^\natural,\cP\circ\cQ_X\circ \cQ_Z^\natural)\nonumber\\&+ \delta(\cP\circ\cQ_X\circ \cQ_Z^\natural,\cC).
\end{align}
By montonicity, the first term on the righthand side is less than $\delta(\cA_X,\cP\circ\cQ_X)$ and therefore
\begin{align}
 		\delta(\cA_X\circ \cQ_Z^\natural,\cC)&\leq \sqrt{2\eps_X(\cA_X)}+\delta(\cP\circ \cQ_X \circ \cQ_Z^\natural,\cC).
 	\end{align}

Next, let $\cF:\sS(\cH_S)\to \cH_{\hat X}$ be the channel with action $\cF(\rho)\mapsto \sum_x p_x \ketbra x$ for some fixed probability distribution $p_x$ and define $\cC=\cP\circ \cF$. 
Again by monotonicity, the second term on the righthand side is less than $\delta(cQ_X\circ\cQ_Z^\natural,\cF)$, and so we obtain
\begin{align}
\min_{\cC}\delta(\cA_X\circ \cQ_Z^\natural,\cC)&\leq \sqrt{2\eps_X(\cA_X)}+\delta(\cQ_X \circ \cQ_Z^\natural,\cF),
 	\end{align}
since we are free to minimize over the maps $\cC$. This expression is equivalent to 
\begin{align}
\sqrt2\eps_X(\cA_X)^{\frac12}+\eta_Z(\cA_X)\geq 1-\delta(\cQ_X \circ \cQ_Z^\natural,\cF).
\end{align}

Now we make use of the dual form, Eq.~\eqref{eq:dual}, of the semidefinite program for the diamond norm to find an upper bound on $\delta(\cQ_X \circ \cQ_Z^\natural,\cF)$.
We first compute $\sC(\cQ_X \circ \cQ_Z^\natural-\cF)$ and then make a suitable choice of $R$. 
Choosing the basis of $\ket\Omega$ to be the $Z$ basis $\ket{\vartheta_z}$, we find that 
 	\begin{align}
 		\sC(\cQ_X \circ \cQ_Z^\natural) &=\sum_{xz}|\langle\varphi_x|\vartheta_z\rangle|^2\ketbra{x}_{\hat X}\otimes\ketbra{\vartheta_z}_S.
 	\end{align}	
For $\sC(\cF)$ we have simply $\sC(\cF)=\sum_x p_x\ketbra{x}_{\hat X}\otimes \id_S$.
Choose $p_x=\frac1d$ and define 
 	\begin{align}
 		R_{\hat X S}=\sum_{xz}\{\tfrac1d-|\langle\varphi_x|\vartheta_z\rangle|^2\}_+\,\ketbra{x}_{\hat X}\otimes\ketbra{\vartheta_z}_S,
 	\end{align}
 	which satisfies the two constraints of \eqref{eq:dual}. We then have
 	\begin{align}
 		\delta(\cQ^\natural_Z\circ \cQ_X,\cF) &\leq \max_z\sum_x \{\tfrac1d-|\langle\varphi_x|\vartheta_z\rangle|^2\}_+.
 	\end{align}
 	Defining $c_2(X;Z)=1-\max_z\sum_x \{\tfrac1d-|\langle\varphi_x|\vartheta_z\rangle|^2\}_+$ completes the proof.
 	\end{proof}
 	
 	Note that the proof comes down to how successive measurement of the two observables acts on the system, much like~\cite{tomamichel_uncertainty_2011}.

\bibliographystyle{apsrev4-1}
\bibliography{InfoDisturbance}
\end{document}